\documentclass[12pt]{article}

\usepackage{mathtools}
\usepackage{booktabs}
\usepackage[english]{babel} 
\usepackage[protrusion=true,expansion=true]{microtype} 
\usepackage{amsmath,amsfonts,amsthm}
\usepackage{ amssymb }
\usepackage{enumitem}
\usepackage{graphicx}
\usepackage{array}
\usepackage[toc,page]{appendix}

\usepackage{booktabs} 
\usepackage{natbib}
\usepackage{setspace}
\usepackage{microtype} 
\usepackage{tabularx}
\usepackage{subcaption}
\usepackage[font = small,labelfont=bf,textfont=it]{caption} 
\usepackage{footnote}
\usepackage{algorithm}
\usepackage[noend]{algpseudocode}
\usepackage{multirow}

\makeatletter
\newcommand{\distas}[1]{\mathbin{\overset{#1}{\kern\z@\sim}}}%
\newsavebox{\mybox}\newsavebox{\mysim}
\newtheorem{theorem}{Theorem}[section]

\newtheorem{proposition}[theorem]{Proposition}
\newtheorem{lemma}[theorem]{Lemma}
\theoremstyle{definition}

\newcommand{\distras}[1]{%
  \savebox{\mybox}{\hbox{\kern3pt$\scriptstyle#1$\kern3pt}}%
  \savebox{\mysim}{\hbox{$\sim$}}%
  \mathbin{\overset{#1}{\kern\z@\resizebox{\wd\mybox}{\ht\mysim}{$\sim$}}}%
}
\newcolumntype{C}[1]{>{\centering\let\newline\\\arraybackslash\hspace{0pt}}m{#1}}

\newcommand{\blind}{1}

\begin{document}

\def\spacingset#1{\renewcommand{\baselinestretch}%
{#1}\small\normalsize} \spacingset{1.45}


\if1\blind
{
	\centering{\bf\Large 
    Calibration for Computer Experiments with Binary Responses and Application to Cell Adhesion Study
    }\\
	  \vspace{0.2in}
} \fi

\if0\blind
{
  \bigskip
  \bigskip
  \bigskip
  \begin{center}
   {\LARGE\bf Calibration for Computer Experiments with Binary Responses and Application to Cell Adhesion Study}
\end{center}
  \medskip
} \fi

\if1\blind
{
	  \centering{Chih-Li Sung$^{a}$ $^{1}$, Ying Hung$^{b}$ \footnote{Joint first authors.},
	  William Rittase$^c$, Cheng Zhu$^c$,\\ C. F. J. Wu$^{d}$ \footnote{Corresponding author.} }\\
	      \vspace{0.2in}
	  \centering{$^a$Department of Statistics and Probability, Michigan State University \\
	  $^b$Department of Statistics, Rutgers, the State University of New Jersey\\
	  $^c$Department of Biomedical Engineering, Georgia Institute of Technology\\
	  $^d$School of Industrial and Systems Engineering, Georgia Institute of Technology\\}
} \fi

\bigskip
\begin{abstract}

Calibration refers to the estimation of unknown parameters which are present in computer experiments but not available in physical experiments. An accurate estimation of these parameters is important because it provides a scientific understanding of the underlying system which is not available in physical experiments.
Most of the work in the literature is limited to the analysis of continuous responses. Motivated by a study of cell adhesion experiments, we propose a new calibration framework for binary responses. Its application to the T cell adhesion data provides insight into the unknown values of the kinetic parameters which are difficult to determine by physical experiments due to the limitation of the existing experimental techniques.

\end{abstract}

\noindent%
\vfill

\newpage
\spacingset{1.45} 

\section{Introduction}
 
To study a scientific problem by experimentation, there are generally two different approaches. 
One is to conduct physical experiments in a laboratory and the other is to perform computer simulations for the study of real systems using mathematical models and numerical tools, such as finite element analysis. Computer experiments have been widely adopted as alternatives to physical experiments, especially for studying complex systems where physical experiments are infeasible, inconvenient, risky, or too expensive. For example, \cite{mak2017efficient} study high-fidelity simulations for turbulent flows in a swirl injector, which are used in a wide variety of engineering applications. In computer experiments, there are two sets of input variables. One is the set of general inputs that represents controllable quantities which are also present in physical experiments, while the other is the set of unknown parameters which represents certain inherent attributes of the underlying systems but cannot be directly controlled or difficult to be measured in physical experiments. These unknown parameters are called \textit{calibration parameters} in the literature \citep{santner2003design}. The focus of this paper is \textit{calibration} which refers to the estimation of the calibration parameters 
using data collected from both physical and computer experiments, so that the computer outputs can closely match the physical responses.
\cite{kennedy2001bayesian} first developed a Bayesian method for calibration and has made a large impact in various fields where computer experiments are used \citep{bayarri2007framework,farah2014bayesian,gramacy2015calibrating, tuo2015efficient,tuo2016theoretical}.

An accurate estimation of calibration parameters is important because it can provide scientific insight that may not be directly obtainable in physical experiments. 
For example, calibration parameters in the implosion simulations are not measurable in physical experiments, the understanding of which provides important information regarding the yield stress of steel and the resulting detonation energy \citep{higdon2008computer}.
In the study of high-energy laser radiative shock system, one of the calibration parameters is the electron flux limiter, which is useful in predicting the amount of heat transferred between cells of a space-time mesh in the simulation but cannot be controlled in physical experiments \citep{gramacy2015calibrating}. 


This paper is motivated by a calibration problem in a study of molecular interactions. We study the molecular interaction by an important type of single molecular experiments called micropipette adhesion frequency assays \citep{chesla1998measuring}. It is the only published method for studying the kinetic rates of cell adhesion, which plays an important role in many physiological and pathological processes. Typically, there are two ways to perform micropipette adhesion frequency experiments: conducting physical experiments in a laboratory, or studying the complex adhesion mechanism by computer simulations based on a kinetic proofreading model 
through a Gillespie algorithm \citep{gillespie1976general}. For both physical and computer experiments, the output is binary, indicating cell adhesion or not \citep{marshall2003direct,zarnitsyna2007memory,huang2010kinetics}.

Binary outputs are common in many applications.
For example, in manufacturing applications, computer simulations are often conducted for failure analysis where the outputs of interest are binary, i.e., failure or success \citep{yan2009reliability}. In other biological problems, binary outputs are observed and evolve in time, such as neuron firing simulations, cell signaling pathways, gene transcription, and recurring diseases \citep{gerstner1997neural,mayrhofer2002devs}.
However, most of the calibration methods are developed for the analysis of continuous outputs. Extensions of the existing calibration methods to binary outputs are not straightforward for two reasons. First, calibration relies on statistical modeling for both computer experiments and physical experiments, but the required modeling techniques for binary outputs are different from those for continuous outputs. Second, the conventional approach to estimating calibration parameters is to match the computer outputs and physical responses, while our interest for binary responses is to match the underlying probability functions for the computer and physical outputs.

To perform calibration in the cell adhesion experiments, we develop a new framework for binary outputs. Calibration parameters are estimated by minimizing the discrepancy between the underlying probability functions in physical experiments and computer experiments. 
The remainder of the paper is organized as follows. Section 2 provides a brief overview of the cell adhesion experiments, including physical experiments and computer simulations. In Section 3, the calibration procedure is described in details.  
Numerical studies are conducted in Section 4 to demonstrate the finite sample performance of estimation. In Section 5, the proposed framework is implemented in the micropipette adhesion experiments. Concluding remarks are given in Section 6. Detailed theoretical proofs are provided in the Appendix and the Supplemental Material.
 
\section{Micropipette Adhesion Frequency Assays}

\subsection{Physical Experiment}

The adaptive immune system defends the organism against diseases by recognition of pathogens by the T cell. T cell receptor (TCR) is the primary molecule on T cell in detecting foreign antigens which are present in major histocompatibility complex (pMHC) molecule expressed by infected cells. Failure to recognize pathogens can result in immune deficiency. False recognition can lead to autoimmune diseases. Therefore, how TCR discriminates different peptides is a central question in the research on adaptive immunity. A micropipette adhesion frequency assay is an important approach to study this TCR-pMHC interaction and mathematically quantify the antigen recognition process.

In a lab, the micropipette adhesion frequency assay is performed as follows (Hung et al. 2008). A red blood cell (RBC, Figure \ref{y:exp}, left) pressurized by micropipette aspiration is used to present the pMHC ligands and to detect binding with the T cell receptor(TCR, Figure \ref{y:exp}, right, only partly shown). The T cell is put into controlled contact with the RBC for a constant area and a preprogrammed duration (Figure \ref{y:exp}B) and then retracted. The output of interest is binary, indicating whether a controlled contact results in adhesion or not. If there is an adhesion between the TCR and pMHC at the end of the contact, retraction will stretch the red blood cell and the RBC membrane will be elongated (Figure \ref{y:exp}C); otherwise the RBC will smoothly restore its spherical shape (Figure \ref{y:exp}A). 

\begin{figure}[ht]
\centering\resizebox{300pt}{100pt} {\includegraphics{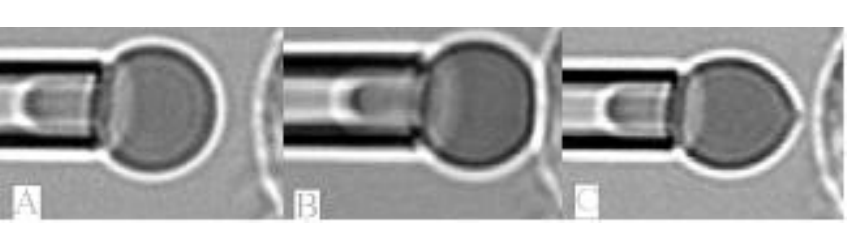}}
\caption{Photomicrographs of micropipette adhesion frequency
assay}\label{y:exp}
\end{figure}

\subsection{Computer Simulation}

Although physical experiments allow accurate measurements of the adhesion frequency, they are time-consuming and often involve complicated experimental manipulation. Moreover, only limited variables of interest can be studied in the lab because of the technical complexity of the biological settings. 
Therefore, a cost-effective approach is to illuminate the unknown biological mechanism in cell adhesion through computer simulations.
Cell adhesion is an intrinsically stochastic process, mathematically stemming from the chemical master equation. 
The basis of its stochasticity includes the 
inherent chemical kinetics of molecules, the quantum indeterminacy in unimolecular reactions, and random perturbations in achieving thermal equilibrium. For cell adhesion, computer simulations can be conducted based on a kinetic proofreading model and simulated through a system of ordinary differential equations (ODEs) governed by the Gillespie algorithm \citep{gillespie1976general}.

Figure \ref{CEplot} illustrates the computer model for the micropipette adhesion frequency assays. At resting state, TCRs in a cluster are considered to be inactive and have on and off rates for pMHC, $\textrm{k}_{\textrm{r}}$ and $\textrm{k}_{\textrm{f}}$, which are unique kinetic parameters to the resting state. Once TCR and pMHC bind, signaling is induced and the cluster of TCRs undergoes one-step kinetic proofreading described by the parameter $\textrm{k}_{\textrm{c}}$. If unsuccessful, nothing happens and the TCR unbinds from the MHC (Major Histocompatibility Complex). If successful, the cluster of TCRs switches states to an upregulated xTCR state, governed by two new kinetic rates, $\textrm{k}_{\textrm{f},\textrm{p}}$ and $\textrm{k}_{\textrm{r},\textrm{p}}$, which are unique to the upregulated state. The cluster of TCRs will then revert back to the resting state after a period of time with a specific half-life parameter, $\textrm{T}_{\textrm{half}}$. Detailed discussions of the simulation can be found in \cite{rittase2018combined}.

In addition to some shared control variables in both physical and computer experiments, several kinetic parameters appear only in the simulation, such as the off-rate enhancement of activated T-cell receptors denoted by $\textrm{k}_{\textrm{r},\textrm{p}}$ in Figure \ref{CEplot}. It is of significant interest to know or estimate the values of these kinetic parameters because they can shed new light on the biological understanding of cell adhesion beyond lab experiments.

\begin{figure}[ht]
\begin{center}
\includegraphics[width=0.7\textwidth]{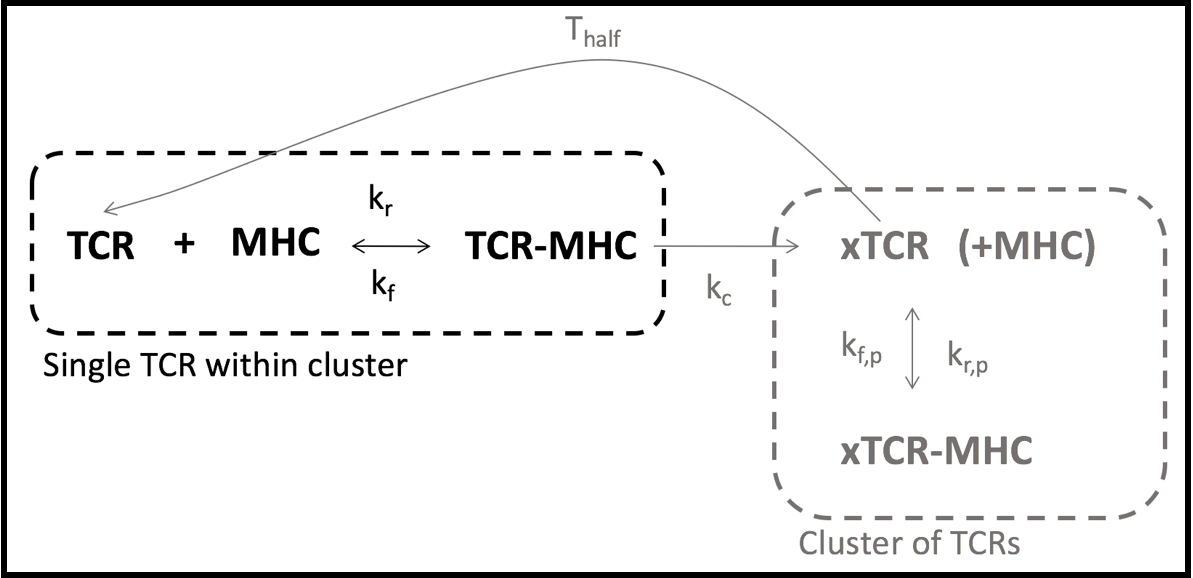}
\caption{Illustration of the computer model.}\label{CEplot}
\end{center}
\end{figure}

\section{A New Calibration Framework}\label{sec:l2calibration}

Suppose $n$ binary outputs are observed from physical experiments and are denoted by $(y^p_1, \cdots, y^p_n)$, where the superscript $p$ stands for ``physical'' and $y^p_i$ is the  $i$th observation taking value 0 or 1. Let $\Omega$ denote a $d$-dimensional experimental region for the control variables $\mathbf{x}$, which is a convex and compact subset of $\mathbb{R}^d$. For each output, the corresponding setting of the control variable is denoted by $\mathbf{x}_i$, where  $i=1,\ldots,n$. Suppose the probability of observing $y=1$ is assumed to have the following model,  
\begin{equation}\label{eq:physicalmodel}
\eta(\mathbf{x}_i)= Pr(y^p_i=1|\mathbf{x}_i)=g(\xi(\mathbf{x}_i)),
\end{equation}
where $g$ is a pre-specified link function. For the binary outputs, we assume $g$ to be the commonly used logistic function, i.e., $g(x)=1/(1+\exp\{-x\})$. The function $\xi(\cdot)$ is  unknown and often called the {\it true process} in the computer experiment literature \citep{kennedy2001bayesian,tuo2015efficient,tuo2016theoretical}.

To study the same scientific problem, a more cost-effective way is to conduct computer simulations (or called computer experiments in this paper). Apart from the control variables $\mathbf{x}$, computer experiments involve calibration parameters, denoted by $\theta$ and $\theta\in \Theta$ which is a compact subset of $\mathbb{R}^q$. These parameters are of scientific interest but their ``true" values are unknown. The binary output from computer experiment is denoted by $y^s(\mathbf{x},\theta)$ with the superscript $s$ standing for ``simulation''. The conditional expectation of $y^s(\mathbf{x},\theta)$ can be written as  
$$p(\mathbf{x},\theta)=Pr(y^s=1|\mathbf{x},\theta),$$ 
where $(\mathbf{x},\theta) \in \Omega\times\Theta$. Even though computer experiments require less experimental manipulation and have smaller cost compared to physical experiments, they can also be computationally intensive (e.g., the $\Lambda$-cold dark matter model in \cite{higdon2013computer} and the high-fidelity simulation in \cite{mak2017efficient}). Therefore, it is not practical to have simulation conducted over the entire experimental region $\Omega\times\Theta$. Instead, the computer experiments are conducted by employing a careful design of experiment, such as space-filling designs \citep{santner2003design}, on a subset of the experimental region.  

The goal of calibration is to search for the setting of the calibration parameters such that the outputs from physical experiments fit as closely as possible to the corresponding outputs of computer experiments.  
This problem is rigorously formulated by \cite{kennedy2001bayesian} in a Bayesian framework. Despite many successful applications using the Bayesian approach (e.g., \cite{higdon2004combining,higdon2008computer,larssen2006forecasting}), recent studies have raised concerns about the \textit{identifiability} issue of the calibration parameters in \cite{kennedy2001bayesian}. See \cite{bayarri2007framework,han2009simultaneous,farah2014bayesian,gramacy2015calibrating}. To tackle this problem, \cite{tuo2015efficient,tuo2016theoretical} propose a frequentist framework based on the method of $L_2$  \textit{calibration}. The idea is to estimate the calibration parameters by minimizing the $L_2$ distance between the physical output and the computer output. It was shown in \cite{tuo2015efficient,tuo2016theoretical} that the calibration method achieves estimation consistency with an optimal convergence rate. 

Although there is a rich literature on calibration, the existing approaches  focus mainly on continuous outputs. 
Inspired by the optimality of the frequentist approach proposed by \cite{tuo2016theoretical}, we  develop a calibration framework for binary outputs using the idea of $L_2$ projection. Ideally, $\theta$ can be obtained by minimizing the discrepancy measured by the $L_2$ distance between the underlying probability functions in the physical and computer experiments. This can be written as 
\begin{equation}\label{L2part1}
\theta^*=\arg\min_{\theta\in\Theta}\|\eta(\cdot)-p(\cdot,\theta)\|_{L_2(\Omega)},
\end{equation}
where the  $L_2$ norm is defined by $\| f \|_{L_2(\Omega)}=(\int_{\Omega} f^2)^{1/2}$. 
The direct calculation of (\ref{L2part1}), however, is not feasible because the true process $\xi(\cdot)$ in \eqref{eq:physicalmodel} is unknown and therefore $\eta(\cdot)$ is unknown. Furthermore, 
$p(\cdot,\cdot)$ is often unknown because computer experiments are usually too complex to admit a closed-form expression.

Instead of solving (\ref{L2part1}) directly, we propose to perform the $L_2$ calibration based on the estimates of $\eta(\cdot)$ and $p(\cdot,\theta)$. First, the true process $\xi(\cdot)$ is estimated by kernel logistic regression, 
that is,
\begin{equation}\label{eq:fitphysical}
\hat{\xi}_n:=\arg\max_{\xi\in\mathcal{N}_{\Phi}(\Omega)}\frac{1}{n}\sum^n_{i=1}\left(y^p_i\log g(\xi(\mathbf{x}_i))+(1-y^p_i)\log(1-g(\xi(\mathbf{x}_i)))\right)+\lambda_n\|\xi\|^2_{\mathcal{N}_{\Phi}(\Omega)},
\end{equation}
where $\|\cdot\|_{\mathcal{N}_{\Phi}(\Omega)}$ is the norm of the reproducing kernel Hilbert space ${\mathcal{N}_{\Phi}(\Omega)}$ generated by a given positive definite reproducing
kernel $\Phi$, and $\lambda_n>0$ is a tuning parameter, which can be chosen by some model selection criterion like cross-validation. 
Kernel logistic regression is chosen here because of its asymptotic properties as shown in Appendix \ref{sec:asym_klr}, which is critical for the development of estimation consistency of calibration parameters and their semiparametric efficiency.
The optimal function \eqref{eq:fitphysical} has the form  $\hat{\xi}_n(\mathbf{x})=\hat{b}+\sum^n_{i=1}\hat{a}_i\Phi(\mathbf{x}_i,\mathbf{x})$, where $\hat{b}$ and $\{\hat{a}_i\}^n_{i=1}$ can be solved by the iteratively re-weighted least squares algorithm.
Detailed discussions can be found in \cite{green1985semi,hastie1990generalized,wahba1994soft,zhu2005kernel}. Based on the estimated true process, we then have $\hat{\eta}_n=g(\hat{\xi}_n)$. 
Because the computer outputs are binary, $p(\cdot,\cdot)$ is not observable and needs to be estimated. Therefore, we assume that a surrogate model $\hat{p}_N(\cdot,\cdot)$
can be constructed as a good approximation to $p(\cdot,\cdot)$ based on $N$ computer outputs, where $N$ is assumed to be larger than $n$ because computer experiments are usually cheaper than physical experiments. 
Theoretically, methods developed for binary classification that satisfy Assumptions C1 and C2, given in the Supplemental Material S1, can be used to estimate $p(.,.)$.  Under some regularity conditions, the emulators constructed by the existing methods, such as 
Gaussian process classification \citep{williams1998bayesian, nickisch2008approximations,sung2017generalized}, satisfy the assumptions and can be employed in this framework. 
Given $\hat{\eta}_n$ and $\hat{p}_N(\cdot,\cdot)$, we are ready to estimate the calibration parameters by minimizing the $L_2$ projection as follows,
\begin{equation}\label{eq:calibration}
\hat{\theta}_n=\arg\min_{\theta\in\Theta}\|\hat{\eta}_n(\cdot)-\hat{p}_N(\cdot,\theta)\|_{L_2(\Omega)}.
\end{equation}  
The calculation of the $L_2$ norm can be approximated by numerical integration methods in practice, such as Monte Carlo integration \citep{caflisch1998monte}, and the minimization problem in \eqref{eq:calibration} can be solved by some optimization methods with respect to the physically plausible domain defined for $\Theta$. For example, in Section \ref{sec:numericalstudy} and \ref{sec:realdata}, $\Theta$'s are assumed to be closed rectangles, so the optimization method of \cite{byrd1995limited} which allows for box constraints can be employed.

Denote 
$
\eta_0(\mathbf{x})=g(\xi_0(\mathbf{x}))
$, where $\xi_0$ is the true process, it can be shown that the fitted model $\hat{\eta}_n$ converges to $\eta_0$ asymptotically and the optimal convergence rate is discussed in Appendix \ref{sec:asym_klr}. Furthermore, the following result shows that $\hat{\theta}_n$ obtained by the $L_2$ calibration in \eqref{eq:calibration} is consistent with the true calibration parameter $\theta^*$ in \eqref{L2part1} and follows an asymptotically normal distribution. The proof is given in Appendix \ref{sec:proofconsistent}.

\begin{theorem}\label{thm:calibration}
Under the regularity assumptions given in Supplemental Material \ref{sec:assumption}, 
we have 
\begin{itemize}   
\item[(i)] 
\begin{equation}
\label{eq:L2thm}
\hat{\theta}_n-\theta^*=2V^{-1}\left(\frac{1}{n}\sum^n_{i=1}(Y^p_i-\eta(\mathbf{x}_i))\frac{\partial p}{\partial\theta}(\mathbf{x}_i,\theta^*)\right)+o_p(n^{-1/2}),
\end{equation}
where $$V=\mathbb{E}\left[\frac{\partial^2}{\partial\theta\partial\theta^T}( \eta(X)-p(X,\theta^*))^2\right].$$
\item[(ii)] $
\sqrt{n}(\hat{\theta}_n-\theta^*)\xrightarrow{d}\mathcal{N}(0,4V^{-1}WV^{-1})$,
where $W$ is positive definite and can be written as
\begin{equation}\label{eq:W}
W=\mathbb{E}\left[\eta(X)(1-\eta(X))\frac{\partial p}{\partial\theta}(X,\theta^*)\frac{\partial p}{\partial\theta^T}(X,\theta^*)\right].
\end{equation}
\end{itemize}
\end{theorem}

In calibration problems, the parameter of interest is a $q$-dimensional calibration parameter $\theta^*$, while the parameter space of model \eqref{eq:physicalmodel} contains an infinite dimensional function space which covers $\xi$. Therefore, the calibration problem is regarded as a semiparametric problem. If a method can reach the highest estimation efficiency for semiparametric problem, we call it \textit{semiparametric efficient}. See \cite{bickel1993efficient} and \cite{kosorok2008} for details. In Supplemental Material \ref{sec:proofsemiparametric}, we show that, similar to its counterpart for continuous outputs \citep{tuo2015efficient}, the proposed method enjoys the semiparametric efficiency.

\section{Numerical Study}\label{sec:numericalstudy}
\subsection{Example with One Calibration Parameter}\label{sec:onedimsimulation}

We start with a simple example where one control variable is shared in both physical and computer experiments and one calibration parameter is involved in the computer experiments. 
Assume that the binary physical outputs are randomly generated from a Bernoulli distribution denoted by $Y^p\sim Ber(\eta(x))$, where 
\[
\eta(x)=\exp\left\{\exp\left(-0.5x\right)\cos\left(3.5{\pi}x\right)\right\}/3
\]
$x\in\Omega=[0,1]$, and $\eta(x)\in[0,1]$ for all $x\in\Omega$. 
Based on \eqref{eq:physicalmodel} where $\eta(x)=g(\xi(x))$ and $g$ is a logistic function, the true process is $\xi(x)=\log(\eta(x)/(1+\eta(x)))$ in the numerical study.
The binary computer outputs  $Y^s$ are randomly generated from $Ber(p(x,\theta))$, where 
\[
p(x,\theta)=\exp\left\{\exp\left(-0.5x\right)\cos\left(3.5{\pi}x(\theta+0.7)\right)\right\}/3
\]
$x\in\Omega=[0,1],\theta\in\Theta=[0,1]$, and $p(x,\theta)\in[0,1]$ for all $x\in\Omega$ and $\theta\in\Theta$. Figure \ref{fig:numericstudy} shows the functions $\eta(x)$ (black lines) and $p(x,\theta)$ with three different calibration parameters (blue dashed lines). 
In this example, the true calibration parameter is $\theta^*=0.3$ because it leads to zero discrepancy between the probability functions in the physical and computer experiments (Figure \ref{fig:numericstudy}(b)). 


\begin{figure}[ht]
\centering
\includegraphics[width=\linewidth]{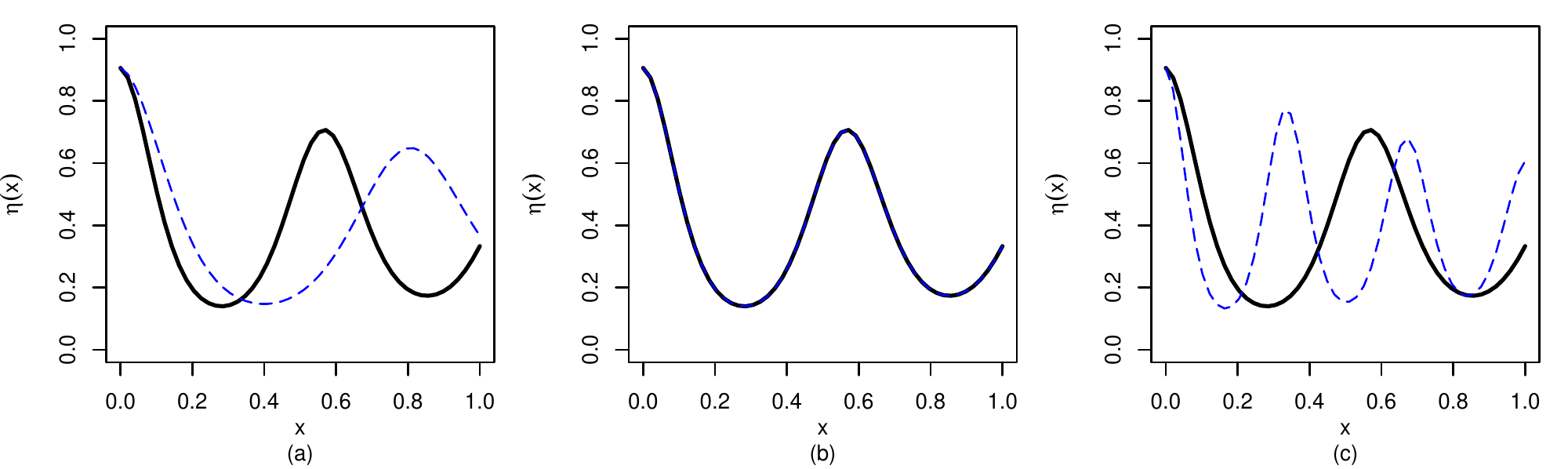}
\caption{True functions in the physical experiment and computer experiment. Black line represents the true function of the physical experiment, and blue line represents the true function of the computer experiment with calibration parameter (a)  $\theta=0$; (b) $\theta=0.3$; (c) $\theta=1$.}
\label{fig:numericstudy}
\end{figure}


Consider the physical experiments with sample size $n$, where the inputs $\{x^p_i\}^{n}_{i=1}$ are selected with equal space in $[0,1]$ and the corresponding outputs are $\{y^p_i\}^n_{i=1}$. The sample size for computer experiments is $N$, the inputs $\{(x^s_i,\theta_i)\}^{N}_{i=1}$ are uniformly selected from $[0,1]^2$ and the corresponding outputs are $\{y^s_i\}^N_{i=1}$. The calibration parameter is estimated by 
\eqref{eq:calibration}, in which $\hat{\eta}_n(x)$ is obtained by the kernel logistic regression \eqref{eq:fitphysical} and the Mat\'ern kernel function \eqref{eq:matern} is chosen with $\nu=2.5$. The tuning parameters $\rho$ and $\lambda_n$ are chosen by cross-validation. The emulator $\hat{p}_N(x,\theta)$ for computer experiments is obtained by the Gaussian process classification proposed by \cite{williams1998bayesian} with the radial basis function kernel 
\begin{equation}\label{eq:RBF}
\Phi_{\phi}((x_i,\theta_i),(x_j,\theta_j))=\exp\left\{-\phi\left((x_i-x_j)^2+(\theta_i-\theta_j)^2\right) \right\},
\end{equation}
where the tuning parameter $\phi$ is chosen by cross-validation in each simulation.

The calibration performance is summarized by the first two rows in Table \ref{tbl:numericstudy} based on 100 replicates. Two combinations of sample sizes, $n$ and $N$, are considered. For each combination, the mean and standard deviation (SD) of the estimated calibration parameters are reported. In general, the proposed method provides reasonable estimation accuracy for the calibration parameter. It also appears that the standard deviation decreases with the increase of sample size. Furthermore, by plugging the true functions $\eta(x)$ and $p(x,\theta)$ and the true calibration parameter $\theta^*=0.3$ into the result in Theorem \ref{thm:calibration}(ii), the asymptotic distribution of  $\hat{\theta}_n$ is $\mathcal{N}(0.3,0.1904/n)$. For the cases of $n=50$ and $100$, the asymptotic standard deviations are 0.0617 and 0.0436, respectively. Not surprisingly, these values are smaller than the empirical standard deviations in Table \ref{tbl:numericstudy} because the estimation uncertainties of $p$, $\eta$, and $\theta^*$ are neglected. The asymptotic distribution of $\hat{\theta}_n$ for $n=50$ is illustrated as the dashed line in Figure \ref{fig:theta_density} and the corresponding empirical distribution is shown as the solid line. It appears that even with a relatively small sample size, the empirical distribution is reasonably close to the asymptotic distribution.

Since there is no existing method that can address calibration problems with binary outputs, the proposed framework is compared with a naive approach that estimates the calibration parameters by minimizing the misclassification rate with respect to $\theta$ at $x^p_i$'s. That is,
\[
\hat{\theta}_{\text{naive}}=\arg\min_{\theta\in\Theta}\sum^n_{i=1}I(y^p_i\neq\hat{y}^s(x^p_i,\theta)),
\]
where $I(\cdot)$ is an indicator function and $\hat{y}^s(x,\theta)$ is the classification rule trained by the data in computer experiments, $\{((x^s_i,\theta_i),y_i^s)\}^N_{i=1}$. Here we consider a random forest classification which provides the best classification performance under the current setting. The calibration results are given by the last two rows of Table \ref{tbl:numericstudy}. According to the results in Table \ref{tbl:numericstudy}, the $L_2$-calibration outperforms the naive method by providing a smaller bias and lower standard deviation in estimation.

\begin{table}[!h]
\centering
\begin{tabular}{
C{2.8cm}C{0.8cm}C{0.8cm}|C{1.5cm}C{1.5cm}C{1.5cm}}
\toprule
Method & $n$ & $N$& $\theta^*$  & Mean & SD \\
\midrule
\multirow{2}{*}{$L_2$-calibration} & 50 & 400 & 0.3 &   0.3050 & 0.1596\\
& 100 & 900 & 0.3 & 0.3169  & 0.1437 \\
\midrule
Naive & 50 & 400 & 0.3 &   0.3730 & 0.2205\\
method & 100 & 900 & 0.3 & 0.3280  & 0.1693 \\
\bottomrule
\end{tabular}
\caption{Mean and standard deviation of the estimated calibration parameters in 100 replicates.}
\label{tbl:numericstudy}
\end{table}

\begin{figure}[ht]
\centering
\includegraphics[width=0.45\linewidth]{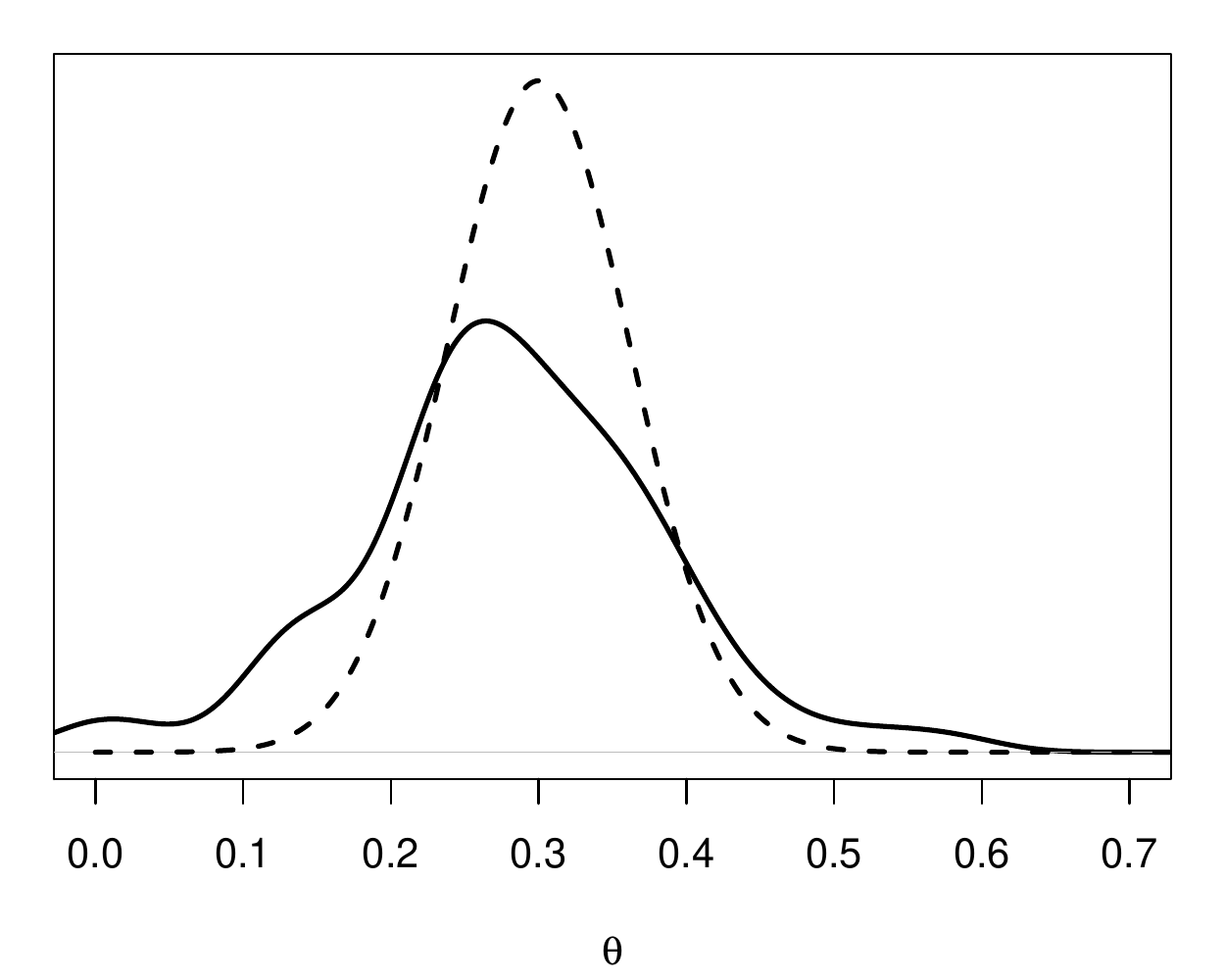}
\caption{The comparison between empirical and asymptotic distributions of the estimates $\hat{\theta}_n$ with $n=50$ and $N=400$. The black line represents the empirical distribution and the dashed line represents the asymptotic distribution.}
\label{fig:theta_density}
\end{figure}

\subsection{Example with Three Calibration Parameters}

In this subsection, we demonstrate the finite sample performance of the proposed method based on an imperfect computer model \citep{tuo2015efficient}, in which there is a discrepancy between physical and computer experiments denoted by $\delta(\mathbf{x})$. Suppose there are three calibration parameters in the computer experiments and two control variables involved in the physical and computer experiments, where $\mathbf{x}\in\Omega=[0,1]^2$. The binary physical outputs are randomly generated from $Y^p\sim Ber(\eta(\mathbf{x}))$, where
\[
\eta(\mathbf{x})=\exp\left\{-4[(2x_1-1)^2+(2x_2-1)^2] \right\}(2x_1-1)+0.65,
\]
which is modified from the 2-dimensional function introduced in \cite{gramacy2008gaussian}, and $\eta(\mathbf{x})\in[0,1]$ for all $\mathbf{x}\in\Omega=[0,1]^2$. Similar to the previous example, the true process can be written as  $\xi(\mathbf{x})=\log(\eta(\mathbf{x})/(1+\eta(\mathbf{x})))$. In the computer experiment, the input variables $(\mathbf{x},\theta)$ are 5-dimensional and the binary computer outputs are randomly generated from $Y^s\sim Ber(p(\mathbf{x},\theta))$, where  
\[
p(\mathbf{x},\theta)=\eta(\mathbf{x})+0.35[(\theta_1-0.3)^2+(\theta_2-0.5)^2+(\theta_3-0.7)^2]+\delta(\mathbf{x}),
\]
$\theta\in\Theta=[0,1]^3$, $p(\mathbf{x},\theta)\in[0,1]$ for all $\mathbf{x}\in\Omega=[0,1]^2$ and $\theta\in\Theta=[0,1]^3$, and $\delta(\mathbf{x})=0.01(x_1-x_2)^2$. By minimizing the $L_2$ distance as in \eqref{L2part1}, we have $\theta^*=(0.3,0.5,0.7)$. Note that the computer model $p(\mathbf{x},\theta)$ is \textit{imperfect} because even with the optimal setting $\theta^*$, there is still a discrepancy between the functions in the computer and the physical experiments.

Similar to the previous example, two combinations of the sample sizes, $n$ and $N$, are considered. The Mat\'ern kernel function \eqref{eq:matern} with $\nu=2.5$ is chosen for fitting $\hat{\eta}_n(\cdot)$, and the tuning parameters in $\hat{\eta}_n(\cdot)$ and $\hat{p}_N(\cdot,\theta)$ are chosen by cross-validation. The estimation results are summarized in Table \ref{tbl:3dnumericstudy2} based on 100 replicates. They show that the proposed method can estimate the three calibration parameters accurately even with an imperfect computer model. The estimation accuracy can be further improved by the increase of sample size, which is similar to the previous example and agrees with the asymptotic results in Theorem \ref{thm:calibration}. 


\begin{table}[!h]
\centering
\begin{tabular}{
C{0.8cm}C{1cm}C{0.8cm}|C{1.5cm}C{1.5cm}C{1.5cm}}
\toprule
$n$ & $N$ &  & $\theta^*$  &  Mean & SD\\
\midrule
\multirow{3}{*}{150} & \multirow{3}{*}{500} & $\hat{\theta}_1$ &0.3 &   0.3606& 0.1794\\
& & $\hat{\theta}_2$ &0.5 &   0.4900 &0.2025 \\
& & $\hat{\theta}_3$ &0.7 &   0.6284 & 0.1763 \\
\midrule
\multirow{3}{*}{250} & \multirow{3}{*}{1500} & $\hat{\theta}_1$ &0.3 &   0.3395 & 0.1711 \\
& & $\hat{\theta}_2$ &0.5 &   0.4956 & 0.2007 \\
& & $\hat{\theta}_3$ &0.7 &   0.6714 & 0.1698 \\
\bottomrule
\end{tabular}
\caption{Mean and standard deviation (SD) of estimated calibration parameters in 100 simulations.}
\label{tbl:3dnumericstudy2}
\end{table}

\section{Analysis of cell adhesion computer experiments}\label{sec:realdata}

In the immune system, it has long been known that T cells utilize their TCR to recognize antigenic pMHC. However, much is still unknown regarding the underlying mechanism.   
To understand the pMHC interactions on T cells, there are two approaches. One is to perform physical experiments in a lab as described in Section 2.1 and the other is to conduct computer simulations as described in Section 2.2.

There are two shared control variables, denoted by $x_{Tc}$ and $x_{Tw}$, in both physical and computer experiments. 
Additionally, four calibration parameters, denoted by $x_{K_{c}},x_{K_{f}},x_{K_{r}}$ and $x_{K_{r,p}}$, only appear in the computer simulations. Their values are of biological interest but cannot be measured or controlled in the lab experiments. The detailed descriptions for these variables are given in Table \ref{realdata}. For the lab experiments, the values of $x_{Tc}$ and $x_{Tw}$ are randomly chosen from the sample space $[0.25,5]\times [1,6]$. The sample size is $n=272$, which is relatively small because it is time-consuming to manipulate different settings in the lab. The design for computer experiments is a 120-run OA-based Latin hypercube design \citep{Tang1993}, and for each run it consists of 10 replicates, i.e., $N=1,200$. To capture the cell-to-cell variability, replications are desirable for physical experiments as well.
However, it is not available in this study due to time and resource constraints. 


\begin{table}[ht] 
\centering
\begin{tabular}{c|c|c}  
\toprule
variable & description & range\\
\midrule
$x_{Tc}$ & cell-cell contact time (second) & [0.25,5]\\
$x_{Tw}$ & waiting time in between contacts (second)&[1,6]\\
\midrule
$x_{Kc}$ & kinetic proofreading rate for activation of cluster (1/second) & [0.1,100]\\
$x_{K_{f}}$ & on-rate enhancement of inactive TCRs ($\mu m^2$/second) &[$10^{-8},10^{-5}$]\\
$x_{K_{r}}$ & off-rate enhancement of inactive TCRs (1/second) & [0.1,10]\\
$x_{K_{r,p}}$ & off-rate enhancement of activated TCRs (1/second) & [0.01,100]\\
\bottomrule
\end{tabular} \caption{Input variables in cell adhesion frequency assay experiments ($x_{Tc}$ and $x_{Tw}$ are control variables and $x_{K_{c}},x_{K_{f}},x_{K_{r}}$ and $x_{K_{r,p}}$ are calibration parameters).}
  \label{realdata}
\end{table}

The physical experiments are analyzed by the kernel logistic regression \eqref{eq:fitphysical} and the fitted model can be written as
\begin{equation*}\label{eq:realphysicalfitting}
\hat{\xi}_n(\mathbf{x})=\hat{b}+\sum^n_{i=1}\hat{a}_i\Phi(\mathbf{x}_i,\mathbf{x}) \quad \mbox{and} \quad \hat{\eta}_n(\mathbf{x})=1/(1+\exp\{-\hat{\xi}_n(\mathbf{x})\}),
\end{equation*}
where $\mathbf{x}=(x_{Tc},x_{Tw})$, $\Phi$ is the Mat\'ern kernel function \eqref{eq:matern} with $\nu=2.5$ and $\rho=0.5$, and $\hat{b}$ and $\{\hat{a}_i\}^n_{i=1}$ are the estimated coefficients. The tuning parameter of the kernel logistic regression is $\lambda_n=0.006$. 
Both of the tuning parameters $\rho$ and $\lambda_n$ are chosen by cross-validation.
The binary data in the physical experiments are plotted in the left panel of Figure \ref{fig:etahat_realdata}. The fitted model $\hat{\eta}_n(\mathbf{x})$ is illustrated in the right panel as a function of contact time and waiting time. 
From biological point of view, the contact duration $x_{Tw}$ is expected to have a positive impact on the adhesion probability because a longer contact period provides a higher chance for T cells and antigen to bind. The impact from waiting time $x_{Tw}$ is expected to be smaller for short contact time and becomes more significant for longer contact time. This is because waiting time is designed for the TCR and antigen to stay in the resting state and avoid a potential memory effect on cell adhesion which is often associated with larger contact duration. This biological information is consistent with the fitted model in Figure \ref{fig:etahat_realdata}.

\begin{figure}[ht]
\centering
\includegraphics[width=0.45\linewidth]{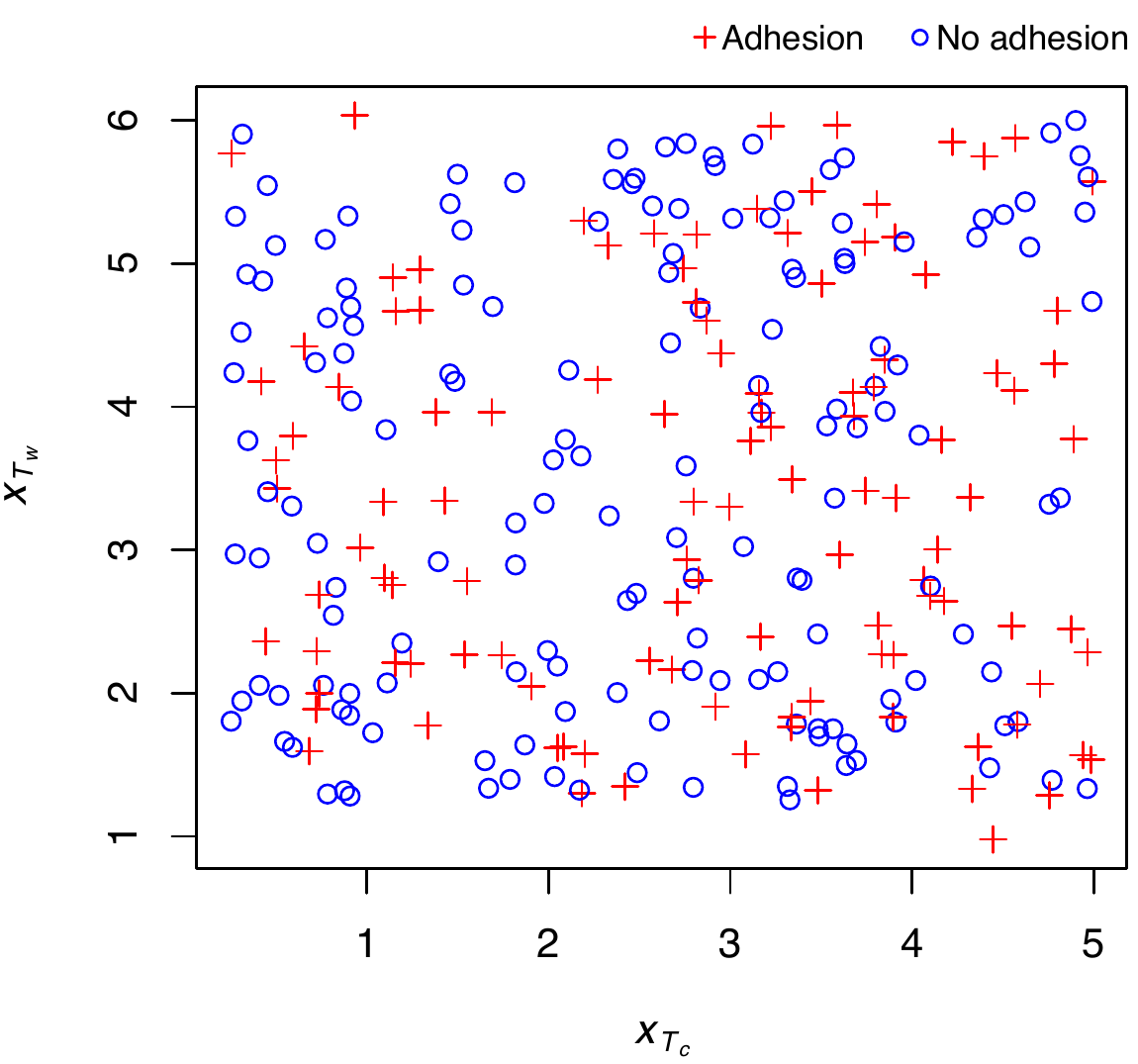}\quad
\includegraphics[width=0.45\linewidth]{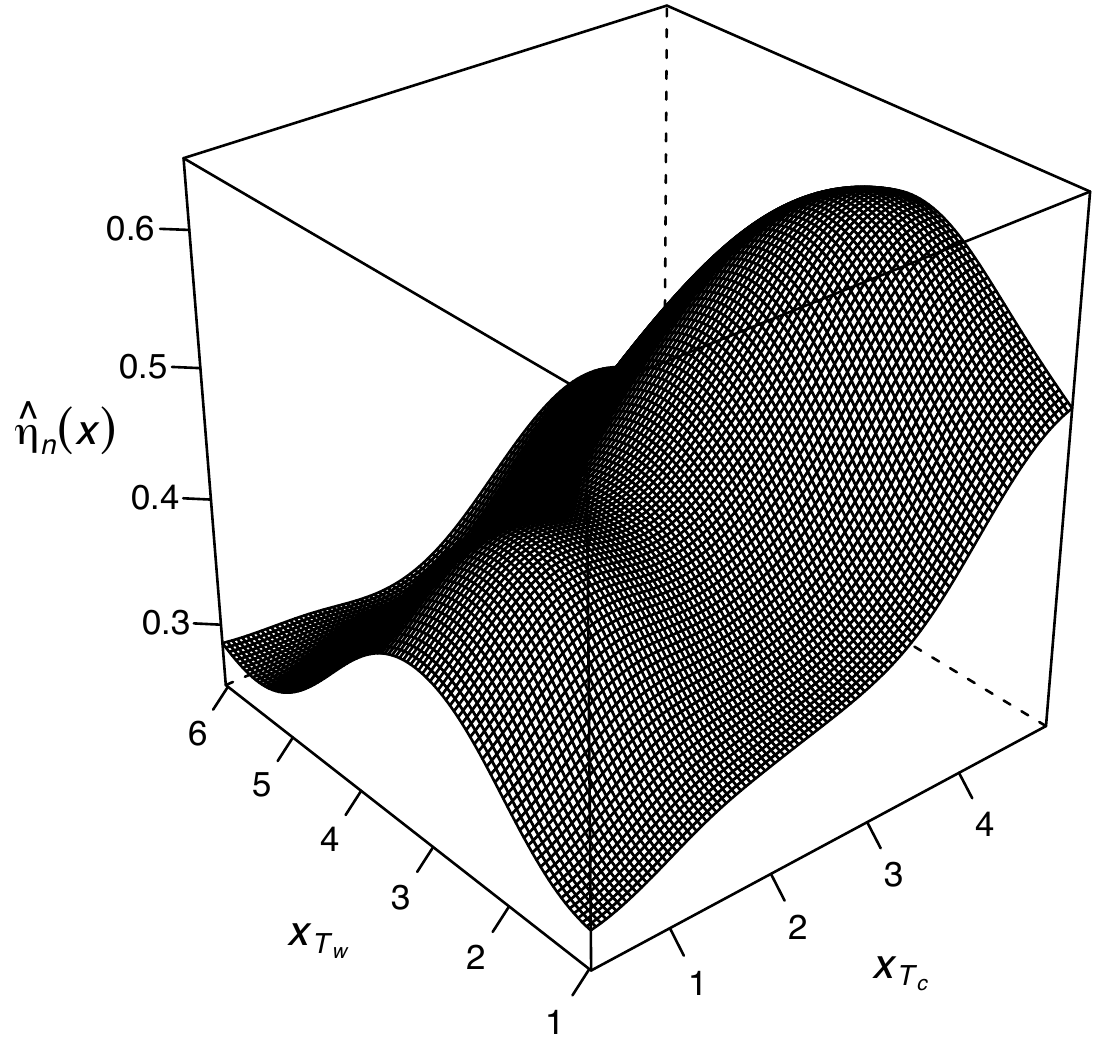}
\caption{Illustration of the data in the physical experiments (left) and the corresponding fitted model (right).}
\label{fig:etahat_realdata}
\end{figure}

The computer experiments are analyzed by the Gaussian process classification proposed by \cite{williams1998bayesian} with the radial basis function kernel \eqref{eq:RBF} with $\phi=21$, which is chosen by cross-validation. The fitted model is $\hat{p}_N(\mathbf{x},\theta)$, where $\mathbf{x}=(x_{Tc},x_{Tw})$ and $\theta= (x_{K_{c}},x_{K_{f}},x_{K_{r}},x_{K_{r,p}})$. 
To gauge the importance of the calibration parameters on adhesion probability, a sensitivity analysis using Monte Carlo estimate of Sobol indices \citep{sobol1993sensitivity} is performed on $\hat{p}_N(\mathbf{x},\theta)$. The sensitivity analysis studies how variable the model output is to changes in the input parameters, and determines which parameters are responsible for the most variation in the model output. We refer more details of sensitivity analysis to \cite{sobol1993sensitivity} and  Chapter 7 of \cite{santner2003design}. The result of the sensitivity analysis is given in Figure \ref{fig:sensitivity}. Each point indicates the estimated Sobol' sensitivity index for each calibration parameter, which measures the proportion of the variation in $\hat{p}_N(\mathbf{x},\theta)$ that is due to the given parameter. The line indicates the corresponding 95\% confidence interval. Figure \ref{fig:sensitivity} shows that all the calibration parameters have impact on the adhesion probability because the estimated indexes are greater than 0, in which $x_{K_f}$ has the highest impact on the adhesion probability.

\begin{figure}[ht]
\centering
\includegraphics[width=0.5\linewidth]{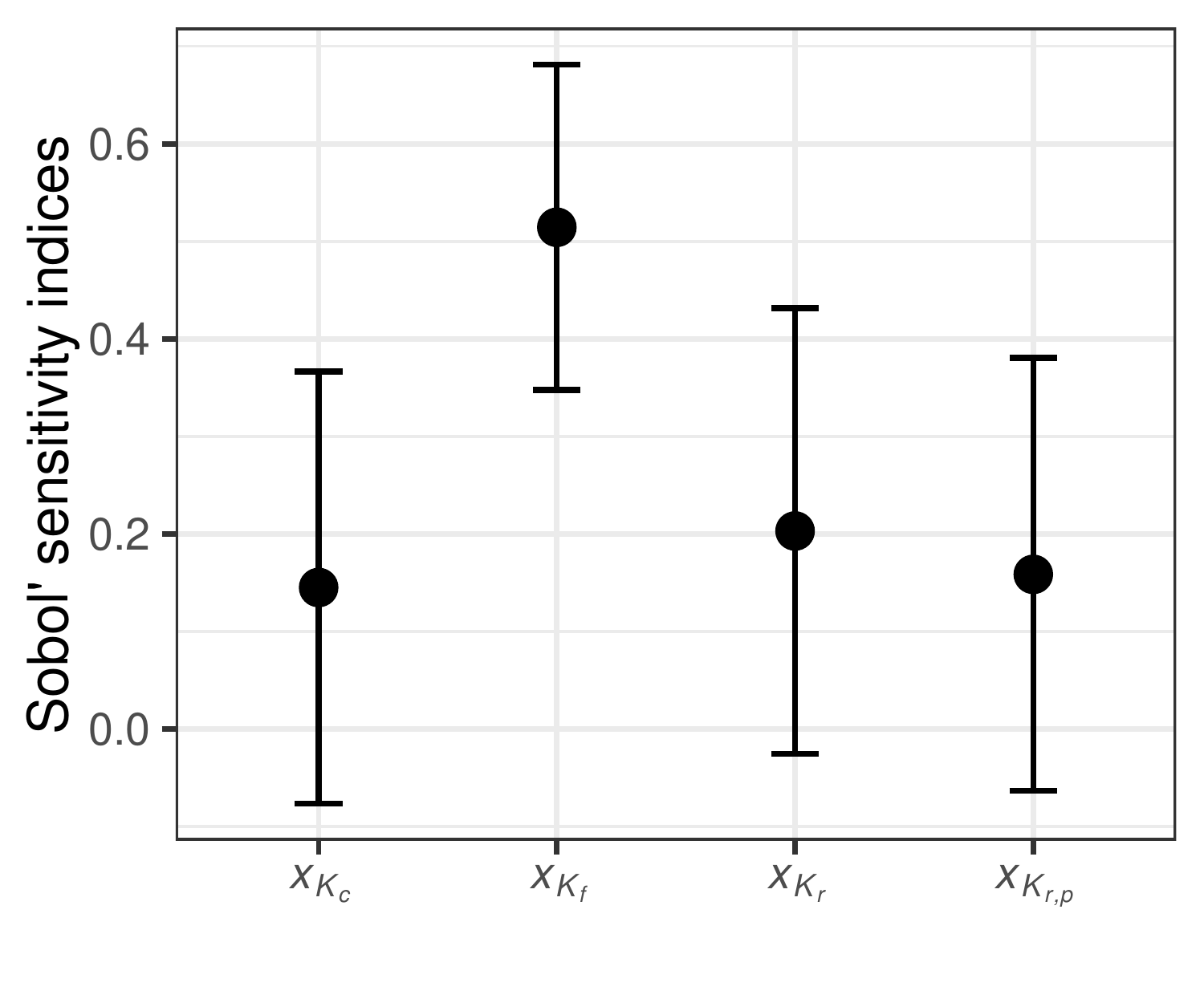}
\caption{Sensitivity analysis using Monte Carlo estimate of Sobol indices \citep{sobol1993sensitivity}. The points are the estimated Sobol' sensitivity indexes for the calibration parameters and the lines are the 95\% confidence intervals.}
\label{fig:sensitivity}
\end{figure}

Based on the two fitted models, the $L_2$ calibration procedure can then be implemented. The estimated calibration parameters are 
\begin{equation}\label{optCAL}
(x_{K_{c}},x_{K_{f}},x_{K_{r}},x_{K_{r,p}})= (3.16,7.77\times 10^{-7}, 0.79, 3.68),
\end{equation}
the corresponding standard deviations calculated based on 
Theorem \ref{thm:calibration} (ii) are (2.15, $3.56\times 10^{-7}$, 0.40, 1.68), and the $L_2$ distance is 0.0461.  
By plugging in the estimated calibration parameter to the emulator $\hat{p}_N(\cdot,\theta)$, the adhesion probabilities obtained from computer experiments (red dashed lines) are compared with those from physical experiments (black lines) in Figure \ref{fig:difference} as a function of the two control variables, contact time and waiting time. It appears that the emulator with the estimated calibration parameters can reasonably capture the trend in the physical experiments.

\begin{figure}
\centering
\includegraphics[width=0.9\linewidth]{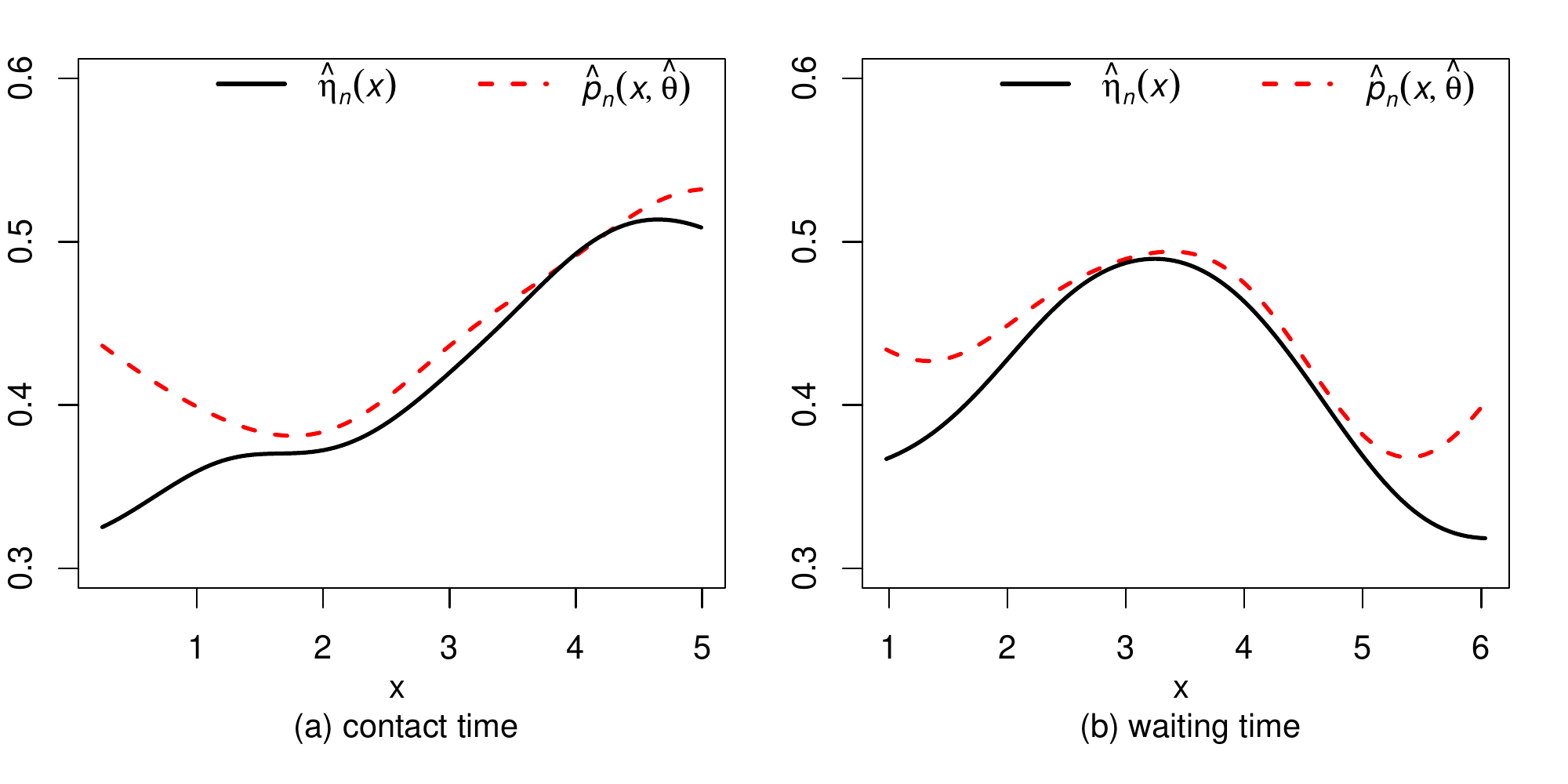}
\caption{Comparison of the fitted model from physical experiments and the fitted model from computer experiments with the calibration parameters given by (\ref{optCAL}). Black lines represent   $\hat{\eta}_n$ and red dash lines represent $\hat{p}_N(\cdot,\hat{\theta})$, where $\hat{\theta}$ is given by  (\ref{optCAL}).}
\label{fig:difference}
\end{figure}

The proposed calibration procedure provides insight into the values of the calibration parameters in the T cell adhesion experiments, which are not available in physical experiments due to the small time scale at which this mechanism operates and the limitation of existing experimental techniques. 
For example, the estimated $x_{K_{f}}$, and $x_{K_{r}}$ are relatively small, which indicates lower on-and-off kinetic rates for the current biological system in the resting state.
The estimated 
$x_{K_{r, p}}$ is larger than $x_{K_{r}}$, which implies a kinetic off-rate increase after TCR binding. Such an increase indicates that the TCR, upon recognizing and binding antigen, quickly releases that antigen, allowing another TCR to rebind. This implies that the mechanism would permit several TCRs to interact with the same antigen in quick succession, which cannot be observed in physical experiments \citep{rittase2018combined}. 




\section{Summary and Concluding Remarks}\label{sec:discussion}
How to estimate the calibrate parameters in computer experiments is an important problem, but the existing calibration methods mainly focus on continuously outputs.  
Motivated by an analysis of single molecular experiments, we propose a new calibration framework for binary responses. The estimate of the calibration parameters is shown to be asymptotically consistent and semiparametric efficient. Our numerical studies confirm the estimation accuracy in finite-sample performance, and the application in single molecular studies illustrates that the proposed calibration method reveals important insight on the underlying adhesion mechanism which cannot be directly observed through existing methods.

Our work lays the foundation for calibration problems with binary responses. This work can be extended in several directions. First, it can be extended to other non-Gaussian data, such as count data. To do so, the logistic function $g$ in \eqref{eq:physicalmodel} can be replaced by other link functions of the exponential family type, such as the log function for Poisson distribution. The true process $\xi$ can then be estimated by maximizing the objective function in \eqref{eq:fitphysical} where the likelihood function of the Bernoulli distribution is replaced by other exponential family distributions. The theoretical results in Section \ref{sec:l2calibration}, however, cannot be directly applied to other exponential family distributions. Moreover, aside from the proposed frequentist framework, the development for Bayesian framework with binary responses is worth exploring. Most of the Bayesian calibration methods suffer from the identifiability issue because the calibration parameters are unidentifiable due to the unknown discrepancy between the true process and the computer model. There are some recent developments, such as the orthogonal Gaussian process \citep{plumlee2017bayesian} and the projected kernel calibration \citep{tuo2017adjustments}, that can address the identifiability issue with continuous outputs. The extension to develop a Bayesian framework for binary outputs is an interesting and important topic for our future research.

It is also worth noting that the optimal setting for calibration parameters is assumed to be unique. But it is possible in practice that this assumption is violated and there are multiple minima with multiple optimal settings for calibration parameters. We suggest two possible solutions to address this problem. First is to compare the optimized calibration parameters with the estimated calibration parameters from other biological systems that share a similar mechanism. In general, these calibration parameters are expected to be consistent and therefore they can be used to validate the calibration results. Another approach is to conduct additional computer simulations based on different optimal calibration settings, and then evaluate their performance by comparing the discrepancy between the simulation outputs and the physical outputs.


\vspace{5mm}
\if1\blind{\noindent \textbf{Acknowledgements}: The authors gratefully acknowledge helpful advice from the associate editor and the referee. This work was supported by NSF DMS 1660504 and 1660477. }
\fi

\bibliography{bib}
\begin{appendices}

\section{Asymptotic Results for Physical Experiment Modeling}\label{sec:asym_klr}

We start with a result developed by \cite{geer2000empirical} for general nonparametric regression (Lemma 11.4 and 11.5 in \cite{geer2000empirical}).
Denote
\begin{equation}\label{eq:truemodel}
\eta_0(\mathbf{x})=g(\xi_0(\mathbf{x})).
\end{equation}
Suppose $\mathcal{F}$ is the class of all regression functions equipped with the Sobolev norm $\|\cdot\|_{H^m(\Omega)}$, which is defined by 
\[
\|\xi\|^2_{H^m(\Omega)}=\|\xi\|^2_{L_2(\Omega)}+\sum^m_{i=1}\left\|\frac{\partial^i \xi}{\partial x^i} \right\|^2_{L_2(\Omega)}.
\]
Let 
\begin{equation*}
\hat{\xi}'_n:=\arg\max_{\xi\in\mathcal{F}}\frac{1}{n}\sum^n_{i=1}\left(y^p_i\log g(\xi(\mathbf{x}_i))+(1-y^p_i)\log(1-g(\xi(\mathbf{x}_i)))\right)+\lambda_n\|\xi\|^2_{H^m(\Omega)},
\end{equation*}
for some $\lambda_n>0$.
Then the convergence rate of $\hat{\xi}'_n$ is given in the following lemma.

\begin{lemma}\label{lemma:generalnp}
Let $\xi_0\in\mathcal{F}$. 
Assume that there exists some nonnegative $k_0$ and $k_1$ so that
\[
k^2_0\leq g(\xi_0(\mathbf{x}))\leq 1-k^2_0\quad\text{and} \quad |\partial g(z)/\partial z|\geq k_1>0\quad\text{for all}\quad |z-z_0|\leq k_1,
\]
where $z_0=\xi_0(\mathbf{x})$ and $\mathbf{x}\in\Omega$. 
For $\lambda^{-1}_n=O(n^{2m/(2m+1)})$, we have 
\[
\|\hat{\xi}'_n\|=O_p(1),\|g(\hat{\xi}'_n)-g(\xi_0)\|_{L_2(\Omega)}=O_p(\lambda^{1/2}_n),
\]
and 
\[
\|\hat{\xi}'_n-\xi_0\|_{L_2(\Omega)}=O_p(\lambda^{1/2}_n).
\]
\end{lemma}

In fact, the norms of some reproducing kernel Hilbert spaces (RKHS) are equivalent to Sobolev norms. For instance, the RKHS generated by the Mat{\'e}rn kernel function, given by
\begin{equation}\label{eq:matern}
\Phi(\mathbf{x},\mathbf{x}')=\frac{1}{\Gamma(\nu)2^{\nu-1}}\left(2\sqrt{\nu}\frac{\|\mathbf{x}-\mathbf{x}'\|}{\rho}\right)^\nu K_\nu\left(2\sqrt{\nu}\frac{\|\mathbf{x}-\mathbf{x}'\|}{\rho}\right),
\end{equation}
where $\nu\geq 1$ and $\rho\in\mathbb{R}_{+}$ are tuning parameters and $K_\nu$ is a Bessel function with parameter $\nu$, is equal to the (fractional) Sobolev space $H^{\nu+d/2}(\Omega)$, and the corresponding norms $\|\cdot\|_{\mathcal{N}_{\Phi}(\Omega)}$ and $\|\cdot\|_{H^{\nu+d/2}(\Omega)}$ are equivalent 
\citep{wendland2004scattered,tuo2016theoretical}. Therefore, as a consequence of Lemma \ref{lemma:generalnp}, we have the following proposition for $\hat{\xi}_n$ obtained by \eqref{eq:fitphysical}.

\begin{proposition}\label{prop:etahat}
Suppose that $\xi_0\in\mathcal{F}=\mathcal{N}_{\Phi}(\Omega)$, and $\mathcal{N}_{\Phi}(\Omega)$ can be embedded into $H^m(\Omega)$. 
Then, for $\lambda^{-1}_n=O(n^{2m/(2m+d)})$, the estimator $\hat{\xi}_n$ in \eqref{eq:fitphysical} and $\hat{\eta}_n=g(\hat{\xi}_n)$ satisfy
\begin{equation*}
\|\hat{\xi}_n\|_{\mathcal{N}_{\Phi}(\Omega)}=O_p(1),\quad\|\hat{\eta}_n-\eta_0\|_{L_2(\Omega)}=O_p(\lambda^{1/2}_n),
\end{equation*}
and
\[
\|\hat{\xi}_n-\xi_0\|_{L_2(\Omega)}=O_p(\lambda^{1/2}_n).
\]
\end{proposition}
Proposition \ref{prop:etahat} suggests that one may choose $\lambda_n\asymp n^{-2m/(2m+d)}$ to obtain the best convergence rate $\|\hat{\eta}_n-\eta_0\|_{L_2(\Omega)}=O_p(n^{-m/(2m+d)})$, where $a_n\asymp b_n$ denotes that the two positive sequences $a_n$ and $b_n$ have the same order of magnitude.

\section{Proof of Theorem \ref{thm:calibration}}\label{sec:proofconsistent}

\begin{proof}
The proof of (i) is developed along the lines described in Theorem 1 of \cite{tuo2015efficient} and under the regularity assumptions A1-A4, B1-B5, C1-C2 in Supplemental Material \ref{sec:assumption}. We first prove the consistency, $\hat{\theta}_n\xrightarrow{p} \theta^*$. It suffices to prove that $\|\hat{\eta}_n(\cdot)-\hat{p}_N(\cdot,\theta)\|_{L_2(\Omega)}$ converges to $\|\eta(\cdot)-p(\cdot,\theta)\|_{L_2(\Omega)}$ uniformly with respect to $\theta\in\Theta$ in probability, which is ensured by
\begin{align}\label{eq:convergenceproof}
&\|\hat{\eta}_n(\cdot)-\hat{p}_N(\cdot,\theta)\|^2_{L_2(\Omega)}-\|\eta(\cdot)-p(\cdot,\theta)\|^2_{L_2(\Omega)}\\
=&\int_{\Omega}\left(\hat{\eta}_n(z)-\eta(z)-\hat{p}_N(z,\theta)+p(z,\theta) \right)\left(\hat{\eta}_n(z)+\eta(z)-\hat{p}_N(z,\theta)-p(z,\theta) \right)dz\nonumber\\
\leq &\left(\|\hat{\eta}_n-\eta \|_{L_2(\Omega)}+\|\hat{p}_N(\cdot,\theta)-p(\cdot,\theta)\|_{L_2(\Omega)} \right)\left( \|\hat{\eta}_n(\cdot)\|_{L_2(\Omega)}+\|\eta(\cdot)\|_{L_2(\Omega)}+\|\hat{p}_N(\cdot,\theta)\|_{L_2(\Omega)}+\|p(\cdot,\theta)\|_{L_2(\Omega)}\right),\nonumber
\end{align}
where the inequality follow from the Schwarz inequality and the triangle inequality. Denote the volume of $\Omega$ by $Vol(\Omega)$. It can be shown that
\[
\|f\|_{L_2(\Omega)}\leq Vol(\Omega)\|f\|_{L_{\infty}(\Omega)}
\]
holds for all $f\in {L_{\infty}(\Omega)}$. Thus, we have 
\begin{align}\label{eq:phat}
\|\hat{p}_N(\cdot,\theta)-p(\cdot,\theta)\|_{L_2(\Omega)}&\leq Vol(\Omega)\|\hat{p}_N(\cdot,\theta)-p(\cdot,\theta)\|_{L_{\infty}(\Omega)}\nonumber \\
&\leq Vol(\Omega)\|\hat{p}_N-p\|_{L_{\infty}(\Omega\times\Theta)},
\end{align}
and $\|f(\cdot)\|_{L_2(\Omega)}\leq Vol(\Omega)$ for $f(\cdot)=\hat{\eta}(\cdot),\eta(\cdot),\hat{p}_N(\cdot,\theta)$, and $p(\cdot,\theta)$ because $\|f(\cdot)\|_{L_{\infty}(\Omega)}\leq 1$. Then, combining \eqref{eq:phat} and assumptions B2 and C1, we have that \eqref{eq:convergenceproof} converges to 0 uniformly with respect to $\theta\in\Theta$, which proves the consistency of $\hat{\theta}_n$.

Since $\hat{\theta}_n$ minimizes \eqref{eq:calibration}, by invoking assumptions A1,A2 and A4, we have 
\begin{align*}
0&=\frac{\partial}{\partial\theta}\|\hat{\eta}_n(\cdot)-\hat{p}_N(\cdot,\hat{\theta}_n)\|^2_{L_2(\Omega)}\\
&=2\int_{\Omega}\left(\hat{\eta}_n(z)-\hat{p}_N(z,\hat{\theta}_n)\right)\frac{\partial\hat{p}_n}{\partial\theta}(z,\hat{\theta}_n)dz,
\end{align*}
and by assumption B2, C1 and C2, it implies 
\begin{equation}\label{eq:temp1}
\int_{\Omega}\left(\hat{\eta}_n(z)-p(z,\hat{\theta}_n)\right)\frac{\partial p}{\partial\theta}(z,\hat{\theta}_n)dz=o_p(n^{-1/2}).
\end{equation}
Let $l(\xi)=\frac{1}{n}\sum^n_{i=1}\left(y^p_i\log g(\xi(\mathbf{x}_i))+(1-y^p_i)\log(1-g(\xi(\mathbf{x}_i)))\right)+\lambda_n\|\xi\|^2_{\mathcal{N}_{\Phi}(\Omega)}$. From \eqref{eq:fitphysical}, we know that $\hat{\xi}_n$ maximizes $l$ over $\mathcal{N}_{\Phi}(\Omega)$. Since $\hat{\theta}_n\xrightarrow{p} \theta^*$ and by assumption A4, $\frac{\partial p}{\partial\theta}(\cdot,\hat{\theta}_n)\in\mathcal{N}_{\Phi}(\Omega)$ with sufficiently large $n$. 
Define $h(z)=\frac{f(z)}{g(z)(1-g(z))}$ and write $\hat{h}_n=h(\hat{\xi}_n)$. Since $h(z)=1$ for any $z\in\mathbb{R}$ when $g$ is a logit function, we have 
\begin{align}\label{eq:sumCDE}
0&=\frac{\partial}{\partial t}l(\hat{\xi}_n(\cdot)+t\frac{\partial p}{\partial\theta_j}(\cdot,\hat{\theta}_n))|_{t=0}\nonumber\\
&=-\frac{1}{n}\sum^n_{i=1}[g(\hat{\xi}_n(\mathbf{x}_i))-g(\xi(\mathbf{x}_i))]\hat{h}_n(\mathbf{x}_i)\frac{\partial p}{\partial\theta_j}(\mathbf{x}_i,\hat{\theta}_n)+\frac{1}{n}\sum^n_{i=1}(Y_i-g(\xi(\mathbf{x}_i)))\hat{h}_n(\mathbf{x}_i)\frac{\partial p}{\partial\theta_j}(\mathbf{x}_i,\hat{\theta}_n)\nonumber\\
&\quad\quad\quad+2\lambda_n<\hat{\xi}_n,\frac{\partial p}{\partial\theta_j}(\mathbf{x}_i,\hat{\theta}_n)>_{\mathcal{N}_{\Phi}(\Omega)}\nonumber\\
&=-\frac{1}{n}\sum^n_{i=1}[g(\hat{\xi}_n(\mathbf{x}_i))-g(\xi(\mathbf{x}_i))]\frac{\partial p}{\partial\theta_j}(\mathbf{x}_i,\hat{\theta}_n)+\frac{1}{n}\sum^n_{i=1}(Y_i-\eta(\mathbf{x}_i))\frac{\partial p}{\partial\theta_j}(\mathbf{x}_i,\hat{\theta}_n)\nonumber\\
&\quad\quad\quad+2\lambda_n<\hat{\xi}_n,\frac{\partial p}{\partial\theta_j}(\mathbf{x}_i,\hat{\theta}_n)>_{\mathcal{N}_{\Phi}(\Omega)}\nonumber\\
&:=C_n+D_n+E_n.
\end{align}

We first consider $C_n$. Let $A_i(f,\theta)=[g(f(\mathbf{x}_i))-g(\xi(\mathbf{x}_i))]\frac{\partial p}{\partial\theta_j}(\mathbf{x}_i,\theta)$ for $(f,\theta)\in\mathcal{N}_{\Phi}(\Omega,\rho)\times \Theta$ for some $\rho>0$. Define the empirical process
\[
E_{1n}(f,\theta)=\frac{1}{\sqrt{n}}\sum^n_{i=1}\left\{A_i(f,\theta)-\mathbb{E}[A_i(f,\theta)]\right\},
\]
where $\mathbb{E}[A_i(f,\theta)]=\int_{\Omega}[g(f(\mathbf{z}))-g(\xi(\mathbf{z}))]\frac{\partial p}{\partial\theta_j}(\mathbf{z},\theta)d\mathbf{z}$.
By assumption B1, $\mathcal{N}_{\Phi}(\Omega,\rho)$ is Donsker. Thus, by Theorem 2.10.6 in \cite{van1996}, $\mathcal{F}_1=\{g(f)-g(\xi):f\in\mathcal{N}_{\Phi}(\Omega,\rho)\}$ is also Donsker because $g$ is a Lipschitz functions. By assumption A4, the class $\mathcal{F}_2=\{\frac{\partial p}{\partial\theta_j}(\cdot,\hat{\theta}_n),\theta\in U\}$ is Donsker. Since both $\mathcal{F}_1$ and $\mathcal{F}_2$ are uniformly bounded, by Example 2.10.8 in \cite{van1996} the product class $\mathcal{F}_1\times\mathcal{F}_2$ is also Donsker. Thus, the asymptotic equicontinuity property holds, which implies that for any $\epsilon>0$ there exists a $\delta>0$ such that 
\[
\limsup_{n\rightarrow\infty}Pr\left(\sup_{\zeta\in\mathcal{F}_1\times \mathcal{F}_2,\|\zeta\|\leq\delta}\left|\frac{1}{\sqrt{n}}\sum^n_{i=1}(\zeta(\mathbf{x}_i)-\mathbb{E}(\zeta(\mathbf{x}_i))) \right|>\epsilon \right)<\epsilon,
\]
where $\|\zeta\|^2:=\mathbb{E}[\zeta(\mathbf{x}_i)^2]$. See Theorem 2.4 of \cite{mammen1997penalized}. This implies that for any $\epsilon>0$ there exists a $\delta>0$ such that 
\begin{equation}\label{eq:equicontinuity}
\limsup_{n\rightarrow\infty}Pr\left(\sup_{f\in\mathcal{N}_{\Phi}(\Omega,\rho),\theta\in U,\|g(f)-g(\xi)\|_{L_2(\Omega)}\leq\delta}\left|E_{1n}(f,\theta) \right|>\epsilon \right)<\epsilon.
\end{equation}
Suppose $\varepsilon>0$ is a fixed value. Assumption B3 implies that there exists $\rho_0>0$ such that $Pr(\|\hat{\xi}\|_{\mathcal{N}_{\Phi}}>\rho_0)\leq\varepsilon/3$. In addition, choose $\delta_0$ to be a possible value of $\delta$ which satisfies \eqref{eq:equicontinuity} with $\epsilon=\varepsilon/3$ and $\rho=\rho_0$. Assumption B2 implies that $Pr(\|g(\hat{\xi}_n)-g(\xi)\|_{L_2(\Omega)}>\delta_0)<\varepsilon/3$. Define 
\[
\hat{\xi}^{\circ}_n=\begin{cases}
   \hat{\xi}_n, & \text{if }\|\hat{\xi}_n\|_{\mathcal{N}_{\Phi}(\Omega)}\leq\rho_0\text{ and } \|g(\hat{\xi}_n)-g(\xi)\|_{L_2(\Omega)}\leq\delta_0,\\
    \xi, & \text{otherwise}.
  \end{cases}
\]
Then, for sufficiently large $n$, we have 
\begin{align*}
Pr(|E_{1n}(\hat{\xi}_n,\hat{\theta}_n)|>\varepsilon)&\leq Pr(|E_{1n}(\hat{\xi}^{\circ}_n,\hat{\theta}_n)|>\varepsilon)+Pr(\|\hat{\xi}_n\|_{\mathcal{N}_{\Phi}(\Omega)}>\rho_0)+Pr(\|g(\hat{\xi}_n)-g(\xi)\|_{L_2(\Omega)}>\delta_0)\\
&\leq Pr(|E_{1n}(\hat{\xi}^{\circ}_n,\hat{\theta}_n)|>\varepsilon/3)+\varepsilon/3+\varepsilon/3\\
&\leq Pr\left(\sup_{f\in\mathcal{N}_{\Phi}(\Omega,\rho),\theta\in U,\|g(f)-g(\xi)\|_{L_2(\Omega)}\leq\delta}\left|E_{1n}(f,\theta) \right|>\varepsilon/3 \right)+\varepsilon/3+\varepsilon/3\\
&\leq\varepsilon.
\end{align*}
The first and third inequalities follow from the definition of $\hat{\xi}^{\circ}_n$, and the last inequality follows from \eqref{eq:equicontinuity}. Thus, this implies that $E_{1n}(\hat{\xi}_n,\theta)$ tends to zero in probability, which gives 
\begin{align*}
o_p(1)&=E_{1n}(\hat{\xi}_n,\hat{\theta}_n)\\
&=\frac{1}{\sqrt{n}}\sum^n_{i=1}\left\{[g(\hat{\xi}_n(\mathbf{x}_i))-g(\xi(\mathbf{x}_i))]\frac{\partial p}{\partial\theta_j}(\mathbf{x}_i,\hat{\theta}_n)\right\}-\frac{1}{\sqrt{n}}\int_{\Omega}[g(\hat{\xi}_n(\mathbf{z}))-g(\xi(\mathbf{z}))]\frac{\partial p}{\partial\theta_j}(\mathbf{z},\hat{\theta}_n)d\mathbf{z}\\
&=-\sqrt{n}C_n-\sqrt{n}\int_{\Omega}[g(\hat{\xi}_n(\mathbf{z}))-g(\xi(\mathbf{z}))]\frac{\partial p}{\partial\theta_j}(\mathbf{z},\hat{\theta}_n)d\mathbf{z},
\end{align*}
which implies
\begin{align}\label{eq:Cn}
C_n
=&-\int_{\Omega}[g(\hat{\xi}_n(\mathbf{z}))-g(\xi(\mathbf{z}))]\frac{\partial p}{\partial\theta_j}(\mathbf{z},\hat{\theta}_n)d\mathbf{z}+o_p(n^{-1/2})\nonumber\\
=&-\int_{\Omega}[\hat{\eta}_n(\mathbf{z})-\eta(\mathbf{z})]\frac{\partial p}{\partial\theta_j}(\mathbf{z},\hat{\theta}_n)d\mathbf{z}+o_p(n^{-1/2}).
\end{align}
Then, by substituting \eqref{eq:temp1} to \eqref{eq:Cn} and using assumption A2, Taylor expansion can be applied to \eqref{eq:Cn} at $\theta^*$, which leads to 
\begin{align*}
C_n=&-\int_{\Omega}[p(\mathbf{z},\hat{\theta}_n)-\eta(\mathbf{z})]\frac{\partial p}{\partial\theta_j}(\mathbf{z},\hat{\theta}_n)dz+o_p(n^{-1/2})\nonumber\\
=&-\left(\frac{1}{2}\int_{\Omega}\frac{\partial^2}{\partial\theta_i\partial\theta_j}[p(\mathbf{z},\tilde{\theta}_n)-\eta(\mathbf{z})]^2d\mathbf{z}\right)(\hat{\theta}_n-\theta^*)+o_p(n^{-1/2}),
\end{align*}
where $\tilde{\theta}_n$ lies between $\hat{\theta}_n$ and $\theta^*$. By the consistency of $\hat{\theta}_n$, we then have $\tilde{\theta}_n\xrightarrow{p} \theta^*$, which implies that 
\begin{equation*}
\int_{\Omega}\frac{\partial^2}{\partial\theta\partial\theta^T}[p(\mathbf{z},\tilde{\theta}_n)-\eta(\mathbf{z})]^2d\mathbf{z}\xrightarrow{p}\int_{\Omega}\frac{\partial^2}{\partial\theta\partial\theta^T}[p(\mathbf{z},\theta^*)-\eta(\mathbf{z})]^2d\mathbf{z}=V.
\end{equation*}
Thus, we have
\begin{equation}\label{eq:finalCn}
C_n=-\frac{1}{2}V(\hat{\theta}_n-\theta^*)+o_p(n^{-1/2}).
\end{equation}

Next, we consider $D_n$. Define the empirical process
\begin{align*}
E_{2n}(\theta)&=\frac{1}{\sqrt{n}}\sum^n_{i=1}\left\{e_i\frac{\partial p}{\partial\theta_j}(\mathbf{x}_i,\theta)-e_i\frac{\partial p}{\partial\theta_j}(\mathbf{x}_i,\theta^*)-\mathbb{E}\left[e_i\frac{\partial p}{\partial\theta_j}(\mathbf{x}_i,\theta)-e_i\frac{\partial p}{\partial\theta_j}(\mathbf{x}_i,\theta^*)\right]\right\}\\
&=\frac{1}{\sqrt{n}}\sum^n_{i=1}\left\{e_i\frac{\partial p}{\partial\theta_j}(\mathbf{x}_i,\theta)-e_i\frac{\partial p}{\partial\theta_j}(\mathbf{x}_i,\theta^*)\right\},
\end{align*}
where $\theta\in U$. Assumption A1 implies that the set $\{\zeta_{\theta}\in C(\mathbb{R}\times\Omega):\zeta_{\theta}(e,\mathbf{x})=e \frac{\partial p}{\partial\theta_j}(\mathbf{x},\theta)-e \frac{\partial p}{\partial\theta_j}(\mathbf{x},\theta^*),\theta\in U\}$ is a Donsker class, which ensures that $E_{2n}(\cdot)$ converges weakly in $L_{\infty}(U)$ to a tight Gaussian process, denoted by $G(\cdot)$. Without loss of generality, we assume $G(\cdot)$ has continuous sample paths. Then, by the continuous mapping theorem \citep{van2000asymptotic} and the consistency of $\hat{\theta}_n$, we have $E_{2n}(\theta)\xrightarrow{p}G(\theta^*)$. Because $E_{2n}(\theta^*)=0$ for all $n$, we have $G(\theta^*)=0$. Then, we have $E_{2n}(\theta)\xrightarrow{p}0$, which gives 
\begin{equation}\label{eq:Dn}
D_n=\frac{1}{n}\sum^n_{i=1}(Y^p_i-\eta(\mathbf{x}_i))\frac{\partial p}{\partial\theta_j}(\mathbf{x}_i,\theta^*)+o_p(n^{-1/2}).
\end{equation}

Lastly, we consider $E_n$. Applying assumption A4, B3, B4, we have
\begin{equation}\label{eq:En}
E_n\leq 2\lambda_n\|\hat{\xi}_n\|_{\mathcal{N}_{\Phi}(\Omega)}\left\|\frac{\partial p}{\partial \theta_j}(\cdot,\hat{\theta}_n) \right\|_{\mathcal{N}_{\Phi}(\Omega)}=o_p(n^{-1/2}).
\end{equation}

By combining \eqref{eq:sumCDE}, \eqref{eq:finalCn}, \eqref{eq:Dn} and \eqref{eq:En}, we have 
\[
\hat{\theta}_n-\theta^*=2V^{-1}\left\{\frac{1}{n}\sum^n_{i=1}(Y^p_i-\eta(\mathbf{x}_i))\frac{\partial p}{\partial\theta}(\mathbf{x}_i,\theta^*)\right\}+o_p(n^{-1/2}).
\]
\end{proof}


\end{appendices}

\def\spacingset#1{\renewcommand{\baselinestretch}%
{#1}\small\normalsize} \spacingset{1.3}

\newpage
\setcounter{page}{1}
\bigskip
\bigskip
\bigskip
\begin{center}
{\Large\bf Supplementary Materials for ``Calibration for Computer Experiments with Binary Responses and Application to Cell Adhesion Study''}
\end{center}
\medskip

\setcounter{section}{0}
\setcounter{equation}{0}
\def\theequation{S\arabic{section}.\arabic{equation}}
\def\thesection{S\arabic{section}}

\section{Assumptions}\label{sec:assumption}

The regularity conditions on the models are given below. For any $\theta\in\Theta\subset\mathbb{R}^q$, write $\theta=(\theta_1,\ldots,\theta_q)$. Denote $e_i=y^p_i-\eta(\mathbf{x}_i)$.

\begin{enumerate}[label=A\arabic*:]
\item The sequences $\{\mathbf{x}_i\}$ and $\{e_i\}$ are independent; $\mathbf{x}_i$'s are i.i.d. from a uniform distribution over $\Omega$; and $\{e_i\}$ is a sequence of i.i.d. random variables with zero mean and finite variance.
\item $\theta^*$ is the unique solution to \eqref{L2part1} and is an interior point of $\Theta$.
\item $V:=\mathbb{E}\left[\frac{\partial^2}{\partial\theta\partial\theta^T}( \eta(X)-p(X,\theta^*))^2\right]$ is invertible.
\item There exists a neighborhood $U\subset \Theta$ of $\theta^*$ such that 
\[
\sup_{\theta\in U}\left\|\frac{\partial p}{\partial \theta_i}(\cdot,\theta)\right\|_{\mathcal{N}_{\Phi}(\Omega)}<+\infty, \quad\frac{\partial^2 p}{\partial \theta_i\partial \theta_j}(\cdot,\cdot)\in C(\Omega\times U),
\]
for all $\theta\in U$ and all $i,j=1,\ldots,q$.
\end{enumerate}

   Assumptions B1-B4 are related to the nonparametric models and Assumptions C1 and C2 are related to the emulators.

\begin{enumerate}[label=B\arabic*:]
\item $\xi\in\mathcal{N}_{\Phi}(\Omega)$ and $\mathcal{N}_{\Phi}(\Omega,\rho)$ is Donsker for all $\rho>0$.
\item $\|\hat{\eta}_n-\eta\|_{L_2(\Omega)}=o_p(1).$
\item $\|\hat{\xi}_n\|_{\mathcal{N}_{\Phi}(\Omega)}=O_p(1).$
\item $\lambda_n=o_p(n^{-1/2})$.
\end{enumerate}
\begin{enumerate}[label=C\arabic*:]
\item $\|\hat{p}_N-p\|_{L_{\infty}(\Omega\times\Theta)}=o_p(N^{-1/2})$.
\item $\|\frac{\partial\hat{p}_N}{\partial\theta_i}-\frac{\partial p}{\partial\theta_i}\|_{L_{\infty}(\Omega\times\Theta)}=o_p(N^{-1/2})$ for $i=1,\ldots,q$.
\end{enumerate} 

The Donsker property is an important concept in the theoretical studies of empirical processes. The definition and detailed discussion are referred to \cite{suppvan1996} and \cite{suppkosorok2008}. \cite{supptuo2015efficient} showed that if the conditions of Proposition \ref{prop:etahat} hold and $m>d/2$, then $\mathcal{N}_{\Phi}(\Omega,\rho)$ is a Donsker. 
The authors also mentioned that under the assumption A1 and $\mathbb{E}[\exp\{C|Y^p_i-\eta(\mathbf{x}_i)|\}]<+\infty$ for some $C>0$, the conditions of Proposition \ref{prop:etahat} are satisfied. Therefore, by choosing a suitable sequence of $\lambda_n$, say $\lambda_n\asymp n^{-2m/(2m+d)}$, one can show that condition B4 holds and B2 and B3 are ensured by Proposition \ref{prop:etahat}.
Assumptions C1 and C2 assume that the   
approximation error caused by emulation in computer experiments is negligible compared to the estimation error caused by the error in physical experiments. Given the fact that the cost for computer experiments is usually cheaper than physical experiments, this assumption is reasonable because the sample size of computer experiments is usually larger than that of physical experiments (i.e., $N>n$). 
Under some regularity conditions, the emulators constructed by the existing methods, such as \cite{suppwilliams1998bayesian}, \cite{suppnickisch2008approximations}, and \cite{suppsung2017generalized},  satisfy the assumptions and can be applied in this framework. 


\section{Theorem and Proof of Semiparametric Efficiency}\label{sec:proofsemiparametric}

\begin{theorem}
Under the Assumptions in Supplemental Material \ref{sec:assumption}, the $L_2$ calibration method \eqref{eq:calibration} is semiparametric efficient.
\end{theorem}

\begin{proof}
If suffices to show that $\hat{\theta}_n$ given in \eqref{eq:calibration} has the same asymptotic variance as the estimator obtained by using maximum likelihood (ML) method.
Consider the following $q$-dimensional parametric model indexed by $\gamma$,
\begin{equation}\label{eq:parametrize}
\xi_{\gamma}(\cdot)=\xi(\cdot)+\gamma^T\frac{\partial p}{\partial\theta}(\cdot,\theta^*),
\end{equation}
with $\gamma\in\mathbb{R}^q$. By combining \eqref{eq:physicalmodel} and \eqref{eq:parametrize}, it becomes a traditional logistic regression model with coefficient $\gamma$. Regarding the model \eqref{eq:physicalmodel}, the true value of $\gamma$ is 0. Hence, under the regularity conditions of Theorem \ref{thm:calibration}, the ML estimator has the asymptotic expression
\begin{equation}\label{eq:MLgamma}
\hat{\gamma}_n=\frac{1}{n}W^{-1}\sum^n_{i=1}(Y_i-\eta(\mathbf{x}_i))\frac{\partial p}{\partial\theta_j}(\mathbf{x}_i,\theta^*)+o_p(n^{-1/2}),
\end{equation}
where $W$ is defined in \eqref{eq:W}. Then a natural estimator for $\theta^*$ in \eqref{L2part1} is
\begin{equation}\label{eq:MLtheta}
\hat{\theta}^{\text{ML}}_n=\arg\min_{\theta\in\Theta}\|\xi_{\hat{\gamma}_n}(\cdot)-p(\cdot,\theta^*)\|_{L_2(\Omega)}.
\end{equation}
Since the ML estimators \eqref{eq:MLgamma} and \eqref{eq:MLtheta} have the same expression as (3.22) and (3.23) in \cite{supptuo2015efficient}, it follows that 
\begin{equation}\label{eq:MLexpression}
\hat{\theta}^{\text{ML}}_n-\theta^*=2V^{-1}\left(\frac{1}{n}\sum^n_{i=1}(Y_i-\eta(\mathbf{x}_i))\frac{\partial p}{\partial\theta}(\mathbf{x}_i,\theta^*)\right)+o_p(n^{-1/2}).
\end{equation}
Therefore, since the asymptotic expression of the ML estimator in \eqref{eq:MLexpression} has the same form as the asymptotic expression of the $L_2$ calibration given by \eqref{eq:L2thm}, the $L_2$ calibration \eqref{eq:calibration} is semiparametric efficient.
\end{proof}

\end{document}